\documentclass[conference]{IEEEtran}
\usepackage[mathscr]{euscript}
\usepackage{amsmath,amssymb,color}
\usepackage{amssymb}

\newtheorem{thm}{Theorem}

\newtheorem{lem}[thm]{Lemma}
\newtheorem{clm}[thm]{Claim}
\newtheorem{dfn}[thm]{Definition}
\newenvironment{proof}
{\noindent\textsc{Proof:}}{\hfill $\square$}

\newcommand{\lra}{\leftrightarrow}
\newcommand{\M}{\texttt{M}}
\newcommand{\R}{\texttt{R}}
\newcommand{\SA}{\texttt{S}}

\newcommand{\LA}{\mathbf{LA}}

\newcommand{\ELA}{\exists\mathbf{LA}}

\newcommand{\quot}{\text{\tt div}}
\newcommand{\rem}{\text{\tt rem}}

\newcommand{\row}{\text{\tt r}}
\newcommand{\col}{\text{\tt c}}
\newcommand{\cond}{\mathtt{cond}}
\newcommand{\entry}{\text{\tt e}}
\newcommand{\Lang}{\mathcal{L}_{\LA}}
\newcommand{\ra}{\rightarrow}
\newcommand{\predicatefont}{\mathrm}
\newcommand{\KMM}{\predicatefont{KMM}}
\newcommand{\Menger}{\predicatefont{Menger}}
\newcommand{\Mincover}{\predicatefont{MinCover}}
\newcommand{\Maxselect}{\predicatefont{MaxSelect}}
\newcommand{\Mincut}{\predicatefont{MinCut}}
\newcommand{\Maxdisj}{\predicatefont{MaxDisj}}
\newcommand{\Cutsize}{\predicatefont{CutSize}}
\newcommand{\cover}{\predicatefont{Cover}}
\newcommand{\select}{\predicatefont{Select}}
\newcommand{\Path}{\predicatefont{Path}}
\newcommand{\Disjoint}{\predicatefont{Disjoint}}
\newcommand{\Collectdisj}{\predicatefont{CollectDisj}}
\newcommand{\Cut}{\predicatefont{Cut}}
\newcommand{\SDR}{\predicatefont{SDR}}
\newcommand{\Perm}{\predicatefont{Perm}}
\newcommand{\Unionprop}{\predicatefont{UnionProp}}
\newcommand{\Hall}{\predicatefont{Hall}}
\newcommand{\Dilworth}{\predicatefont{Dilworth}}

\newcommand{\Maxantichain}{\predicatefont{MaxAntiChain}}
\newcommand{\Minchain}{\predicatefont{MinChain}}

\newcommand{\Chain}{\predicatefont{Chain}}

\begin{document}
\title{Feasible combinatorial matrix theory \\
{\large Polytime proofs for K{\"o}nig's Min-Max and related theorems}}

\author{\IEEEauthorblockN{Ariel Fern{\'a}ndez}
\IEEEauthorblockA{McMaster University\\
Hamilton, Canada\\
{\tt fernanag@mcmaster.ca}}
\and
\IEEEauthorblockN{Michael Soltys}
\IEEEauthorblockA{McMaster University\\
Hamilton, Canada\\
{\tt soltys@mcmaster.ca}}}

\maketitle

\begin{abstract}
We show that the well-known K{\"o}nig's Min-Max Theorem (KMM), a
fundamental result in combinatorial matrix theory, can be proven in
the first order theory $\LA$ with induction restricted to $\Sigma_1^B$
formulas.  This is an improvement over the standard textbook proof of
KMM which requires $\Pi_2^B$ induction, and hence does not yield
feasible proofs --- while our new approach does.  
$\LA$ is a weak theory that essentially captures the
ring properties of matrices; however, equipped with $\Sigma_1^B$
induction $\LA$ is capable of proving KMM, and a host of other
combinatorial properties such as Menger's, Hall's and Dilworth's
Theorems. Therefore, our result formalizes Min-Max type of reasoning
within a feasible framework.
\end{abstract}

%
\IEEEpeerreviewmaketitle

\section{Introduction}

In this paper we are concerned with the complexity of formalizing
reasoning about combinatorial matrix theory.  We are interested in the
strength of the bounded arithmetic theories necessary in order to
prove the fundamental results of this field. We show, by introducing
new proof techniques, that the logical theory $\LA$ with induction
restricted to bounded existential matrix quantification is sufficient
to formalize a large portion of combinatorial matrix theory.

Perhaps the most famous theorem in combinatorial matrix theory is the
K{\"o}nig's Mini-Max Theorem (KMM) which arises naturally in all areas
of combinatorial algorithms --- for example ``network flows'' with
``min-cut max-flow'' type of reasoning.
See~\cite{konig-hungarian-1916,konig-german-1916} for the original
papers introducing KMM, and see~\cite{cook-dai-2012} for recent work
related to formalizing proof of correctness of the Hungarian
algorithm, which is an algorithm based on KMM. As far as we know, we
give the first feasible proof of KMM.

As KMM is a cornerstone result, it has several counter-parts in
related areas of mathematics: Menger's Theorem, counting disjoint
paths; Hall's Theorem, giving necessary and sufficient conditions for
the existence of a ``system of distinct representatives'' of a
collection of sets; Dilworth's Theorem, counting the number of
disjoint chains in a poset, etc. We note that we actually show the
equivalence of KMM with a {\em restricted} version of Menger's
Theorem.

We show that KMM can be proven feasibly (Theorem~\ref{thm:main}), and
we do so with a new proof of KMM that relies on introducing a new
notion (Definition~\ref{dfn:dc}).  Furthermore, we show that the
theorems related to KMM, and listed in the above paragraph, can also
be proven feasibly; in fact, all these theorems are equivalent to KMM,
and the equivalence can be shown in $\LA$
(Theorem~\ref{thm:equivalences}). We believe that this captures the
proof complexity of Min-Max reasoning.

Our results show that Min-Max reasoning can be formalized with uniform
Extended Frege. It would be very interesting to know whether the
techniques recently introduced by \cite{Tzameret-2011} could bring the
complexity further down to quasi-polynomial Frege.

\section{Background}\label{sec:II}

KMM states the following: Let $A$
be a matrix of size $n\times m$ with entries in $\{0,1\}$, what we
sometimes call a 0-1 matrix. A {\em line} of $A$ is an entire
row or column of $A$; given an entry $A_{ij}$ of $A$ (when giving
$\LA$ formulas we shall denote such an entry with $A(i,j)$), we say that a
line {\em covers} that entry if this line is either row $i$ or column
$j$. Then, the minimal number of lines in $A$ that
cover all of the 1s in $A$ is equal to the maximal number of 1s in
$A$ with no two of the 1s on the same line. Note that KMM is stated
for $n\times m$ matrices, but for simplicity we shall state it for
$n\times n$ (i.e., square) matrices; of course, all results given in this
paper hold for both.

See~\cite[pg.~6]{brualdi} for a classical discussion and proof of KMM
(and note that the proof relies, implicitly, on a $\Pi_2^B$ type of
induction).

We give a feasible proof of KMM in the logical theory $\LA$ defined
in~\cite{solcoo}. By restricting the induction to be over $\Sigma_1^B$
formulas, that is formulas whose prenex form consists of a block of
bounded existential matrix quantifiers --- and no other matrix
quantifiers --- we manage to prove KMM in a fragment of $\LA$ called
$\ELA$. While the matrices in the statements of the theorems have
$\{0,1\}$ entries, we assume that the underlying ring is $\mathbb{Z}$,
the set of integers. We require the integers as one of our fundamental
operations will be counting the number of~1s in a 0-1 matrix, i.e.,
computing $\Sigma A$, the sum of all the entries of $A$.

The background for $\LA$ is given in a one-page Appendix at the end of
the paper, but the interested reader can read the full treatment
in~\cite{solcoo}.  We shall leave routine details of proofs to the
reader, in the interest of space.

The main contribution of this paper is to show that KMM can be proven
in the theory $\ELA$ which implies that it can be proven feasibly. We
mention here an important observation of Je{\v{r}}{\'a}bek
from~\cite[pg.\ 44]{phd-jerabek}: $\ELA$ does not necessarily
translate into a polytime proof system (e.g., extended Frege) when the
matrices are over $\mathbb{Z}$. However, if we restrict the quantified
matrices to be over $\{0,1\}$, which is what we do in our proofs, it
readily translates into extended Frege.  

We use $|A|\le n$ to abbreviate $r(A)\le n\wedge c(A)\le n$, that is,
the number of rows of $A$ is bounded by $n$, and the number of columns
of $A$ is bounded by $n$.  We let $(\exists A\le n)\alpha$ --- resp.\
$(\forall A\le n)\alpha$ --- abbreviate $(\exists A)[|A|\le
n\wedge\alpha]$ --- resp.\ $(\forall A)[|A|\le n\ra\alpha]$.  These
are {\em bounded matrix quantifiers}.

Note that $\LA$ allows for reasoning with arbitrary quantification;
however, in $\LA$ we only allow induction over formulas without matrix
quantification.  On the other hand, in $\ELA$ we allow induction over
so called $\Sigma_1^B$ formulas.  These are formulas, which when
presented in prenex form, contain a single block of bounded
existential matrix quantifiers.  The set of formulas $\Pi_1^B$ is
defined similarly, except the block of quantifiers is universal. 

In general, $\Sigma_i^B$ is the set of formulas which, when presented
in prenex form, start with a block of bounded existential matrix
quantifiers, followed by a block of bounded universal matrix
quantifiers, with $i$ such alternating blocks. The set $\Pi_i^B$ is
the same, except it starts with a block of universal matrix
quantifiers.

The main contribution of our paper can now be stated more precisely:
following K{\"o}nig's original proof of KMM, which is also the
standard presentation of the proof in the literature (see the seminal
work in the field \cite[pg.\ 6]{brualdi}) one can construct a proof
with $\Pi_2^B$ induction, which does not yield translations into extended
Frege proofs.  On the other hand, we are able to give a proof that
uses only $\Sigma_1^B$ induction, which do yield extended Frege
proofs, and thereby a feasible proof of KMM.  Our
insight is that while we are doing induction over the size of
matrices, we can pre-arrange our matrices in a way that lowers the
complexity of the induction. This is accomplished with the procedure
outlined in the Definition~\ref{dfn:dc} --- the diagonal property for
matrices --- and the subsequent proof of Claim~\ref{clm:2}.

We show how to express the concepts necessary to state KMM in the
language $\Lang$. First, we say that the matrix $\alpha$ is a {\em
cover} of a matrix $A$ with the predicate:
\begin{equation}\label{eq:defcover}
\begin{split}
& \cover(A,\alpha):= \\
& \forall i,j\le
r(A)(A(i,j)=1\ra\alpha(1,i)=1\vee\alpha(2,j)=1)
\end{split}
\end{equation}
We allow some leeway with notation: $\forall i,j\le r(A)$ is of course
$(\forall i\le r(A))(\forall j\le r(A))$.
The matrix $\alpha$ keeps track of the lines that cover $A$; it does
so with two rows: the top row keeps track of the horizontal lines,
and the bottom row keeps track of the vertical line. The condition
ensures that any 1 in $A$ is covered by some line stipulated in
$\alpha$.

The next predicate expresses that the matrix $\beta$ is a selection of
1s of $A$ so that no two of these lines are on the same line.  Thus,
$\beta$ can be seen as a ``subset'' of a permutation matrix; that is,
each $\beta$ is obtained from some permutation matrix by deleting some
(possibly none) of the 1s.  We say that $\beta$ is a {\em selection}
of $A$, and it is given with the following formula:
\begin{equation}\label{eq:defselect}
\begin{split}
\select&(A,\beta):= \forall i,j\le r(A)((\beta(i,j)=1\ra A(i,j)=1) \\
\wedge&\forall k\le r(A)(\beta(i,j)=1\ra\beta(i,k)=0\wedge \beta(k,j)=0))
\end{split}
\end{equation}

We are interested in a minimum cover (as few 1s in $\alpha$ as
possible) and a maximum selection (as many 1s in $\beta$ as possible).
The following two predicates express that $\alpha$ is a minimum cover
and $\beta$ a maximum selection.
\begin{align}
&\Mincover(A,\alpha) := \\
&\cover(A,\alpha)\wedge\nonumber
\forall\alpha'\le c(\alpha)(\cover(A,\alpha')\ra\Sigma\alpha'\ge
\Sigma\alpha)\label{eq:defmincover} \\
&\Maxselect(A,\beta) := \\
&\select(A,\beta)\wedge\nonumber
\forall\beta'\le r(\beta)(\select(A,\beta')\ra\Sigma\beta'\le
\Sigma\beta)\label{eq:defmaxselect}
\end{align}

Clearly $\Mincover$ and $\Maxselect$ are $\Pi_1^B$ formulas.
We can now state KMM in the language of $\Lang$ as follows:
\begin{equation}\label{eq:KMM}
\Mincover(A,\alpha)\wedge\Maxselect(A,\beta)\ra
\Sigma\alpha=\Sigma\beta
\end{equation}
Note that~(\ref{eq:KMM}) is a $\Sigma_1^B$ formula. The reason is that
in prenex form, the universal matrix quantifiers in $\Mincover$ and
$\Maxselect$ become existential as we pull them out of the implication; 
they are also bounded.

Let $\KMM(A,n)$ be the following $\Sigma_1^B$ formula: it is a
conjunction of the statement that $A$ is an $n\times n$ matrix, which
we abbreviate informally as $|A|=n$, and which in $\Lang$ is stated as
$r(A)=n\wedge c(A)=n$, and~(\ref{eq:KMM}) in prenex form:
\begin{equation}\label{eq:KMM2}
\KMM(A,n):=|A|=n\wedge\text{prenex}(\text{\ref{eq:KMM}}).
\end{equation}

Given a matrix $A$, let $l_A$ and $o_A$ denote the minimum number of
lines necessary to cover all the 1s of $A$, and the maximum number of
1s no two on the same line, respectively. (Of course, K{\"o}nig's
theorem says that for all $A$, $l_A=o_A$.) In terms of the definitions
just given, we have that $l_A=\Sigma\alpha$ where
$\Mincover(A,\alpha)$, and $o_A=\Sigma\beta$ where
$\Maxselect(A,\beta)$.

Finally, the fact that $P$ is a permutation matrix can be stated
easily with a predicate free of matrix quantification; see, for
example, \cite{soltys-tcs04}.

\section{The main result}

With the basic machinery in place, we can now prove the main Theorem
of the paper.

\begin{thm}\label{thm:main}
$\ELA\vdash\KMM$.
\end{thm}

What this Theorem says is that $\KMM$ can be shown in $\LA$ with
$\Sigma_1^B$ induction, and thus in uniform extended Frege, which in
turn means feasibly. The rest of this section consists in a proof of
this theorem.

Some of the intermediate results can be shown with just $\LA$
induction (i.e., induction over formulas without matrix quantifiers,
that is, over formulas in $\Sigma_0^B=\Pi_0^B$).  We use the weaker
theory whenever possible.

\begin{lem}\label{lem:1}
Given a matrix $A$, and given
any permutation matrix $P$, we have 
\begin{itemize}
\item $\LA\vdash l_{PA}=l_{AP}=l_A$
\item $\LA\vdash o_{PA}=o_{AP}=o_A$
\end{itemize}
That is, these four equalities can be proven in $\LA$, i.e., with
induction restricted to formulas without matrix quantifiers.
\end{lem}

\begin{proof}
$\LA$ shows that if we reorder the rows or columns (or both) of a
given matrix~$A$, then the new matrix, call it~$A'$, where $A'=PA$ or
$A'=AP$, has the same size minimum cover and the same size maximum
selection. Of course, we can reorder both rows and columns by applying
the statement twice: $A'=PA$ and $A''=A'Q=PAQ$.

The first thing that we need to show is that:
\begin{itemize}
\item $\LA \vdash \cover(A,\alpha)\ra \cover(A',\alpha')$
\item $\LA \vdash \select(A,\beta)\ra \select(A',\beta')$
\end{itemize} 
where $A'$ is defined as in the above paragraph, and $\alpha'$ is the
same as $\alpha$, except the first row of $\alpha$ is now reordered by
the same permutation $P$ that multiplied $A$ on the left (and the
second row of $\alpha$ is reordered if $P$ multiplied $A$ on the
right). The matrix $\beta$ is even easier to compute, as
$\beta'=P\beta$ if $A'=PA$, and $\beta'=\beta P$ if $A'=AP$. It
follows from $P$ being a permutation matrix that
$\Sigma\alpha=\Sigma\alpha'$ and $\Sigma\beta=\Sigma\beta'$: we can
show by $\LA$ induction on the size of matrices that if $X'$ is the
result of rearranging $X$ (i.e., $X'=PXQ$, where $P,Q$ are permutation
matrices), then $\Sigma X=\Sigma X'$. We do so first on $X$ consisting
of a single row, by induction on the length of the row. Then we take
the single row as the basis case for induction over the number of rows
of a general $X$.

It is clear that given $A'$, the cover $\alpha'$ has been adjusted
appropriately; same for the selection $\beta'$. We can prove it
formally in $\LA$ by contradiction: suppose some 1 in $A'$ is not
covered in $\alpha'$; then the same 1 in $A$ would not be covered by
$\alpha$. For the selections, note that reordering rows and/or columns
we maintain the property of being a selection: we can again prove this
formally in $\LA$ by contradiction: if $\beta'$ has two 1s on the same
line, then so would $\beta$.

The next thing to show is that
\begin{itemize}
\item $\LA \vdash \Mincover(A,\alpha)\ra \Mincover(A',\alpha')$
\item $\LA \vdash \Maxselect(A,\beta)\ra \Maxselect(A',\beta')$
\end{itemize} 
and the reasoning that accomplishes this is by contradiction. As
permuting only reorders the matrices (it does not add or take away
1s), if the right-hand side does not hold, we would get that the
left-hand side does not hold by applying the inverse of the
permutation matrix.

All these arguments can be easily formalized in $\LA$, and we leave
the details to the reader.
\end{proof}

The next definition is a key concept in the $\Sigma_1^B$ proof of KMM.

\begin{dfn}\label{dfn:dc}
We say that an $n\times n$ matrix over $\{0,1\}$ has the {\em diagonal
property} if for each diagonal entry $(i,i)$ of $A$, either
$A_{ii}=1$, or $(\forall j\ge i)[A_{ij}=0\wedge A_{ji}=0]$.
\end{dfn}

\begin{clm}\label{clm:2}
Given any matrix $A$, $\ELA$ proves that there exist permutation
matrices $P,Q$ such that $PAQ$ has the diagonal property.
\end{clm}

\begin{proof}
We construct $P,Q$ inductively on $n=|A|$. Let the {\em $i$-th layer
of $A$} consist of the following entries of $A$: $A_{ij}$, for
$j=i,\ldots,n$ and $A_{ji}$ for $j=i+1,\ldots,n$. Thus, the first
layer consists of the first row and column of $A$, and the $n$-th
layer (also the last layer), is just $A_{nn}$.  We transform $A$ by
layers, $i=1,2,3,\ldots$. At step $i$, let $A'$ be the result of
having dealt already with the first $i-1$ layers. If $A'_{ii}=1$ move
to the next layer, $i+1$. Otherwise, find a $1$ in layer $i$ of $A'$.
If there is no 1, also move on to the next layer, $i+1$. If there is a
1, permute it from position $A_{ij'}$, $j'\in\{i,\ldots,n\}$ to
$A'_{ii}$, or from position $A_{j'i}$, $j'\in\{i+1,\ldots,n\}$. Note
that such a permutation does not disturb the work done in the previous
layers; that is, if $A'_{kk}$, $k<i$, was a 1, it continues being a 1,
and if it was not a 1, then there are no 1s in layer $k$ of $A'$.
\end{proof}

It is Claim~\ref{clm:2} that allows us to bring down the complexity of
the proof of KMM from $\Pi_2^B$ to $\Sigma_1^B$.  As we shall see, by
transforming $A$ into $A'$, so that $A'=PAQ$ where $P,Q$ are
permutation matrices and $A'$ has the diagonal form, we can prove KMM
for $A'$ with just $\Sigma_1^B$ induction, and then by
Lemma~\ref{lem:1} we obtain an $\ELA$ proof of KMM for $A$.  All of
this is made precise in the following Lemma; recall that $\KMM(A,n)$
is defined in~(\ref{eq:KMM2}).

\begin{lem}\label{lem:4}
$\ELA\vdash\forall n\KMM(A,n)$.
\end{lem}

We are going to prove Lemma~\ref{lem:4} by induction on $n$, breaking
it down into Claims~\ref{clm:3} and~\ref{clm:4}. Once we have that
$\forall n\KMM(A,n)$, we replace $n$ with $|A|$, and obtain an $\ELA$
proof of~(\ref{eq:KMM}), and thereby a proof of
Theorem~\ref{thm:main}.

From Claims~\ref{lem:1} and~\ref{clm:2} we know that it is sufficient
to prove Lemma~\ref{lem:4} for appropriate $PAQ$, which ensures the
diagonal property spelled out in Claim~\ref{clm:2}.  Thus, in order to
simplify notation, we assume that our $A$ is the result of applying
the permutations; i.e., $A$ has the diagonal property.

\begin{clm}\label{clm:3}
$\LA\vdash o_A\le l_A$.
\end{clm}

\begin{proof}
Given a covering of $A$ consisting of $l_A$ lines, we know that
every~1 we pick for a maximal selection of~1s has to be on one of the
lines of the covering. We also know that we cannot pick more than
one~1 from each line. Thus, the number of lines in the covering
provide an upper bound on the size of such selection, giving us
$o_A\le l_A$. 

We can formalize this argument in $\LA$ as follows: let $A'$ be a
matrix whose rows represent the $l_A$ lines of a covering, and whose
columns represent the $o_A$ 1s no two on the same line.  Let
$A'(i,j)=1\iff$ the line labeled with~$i$ covers the~1 labeled
with~$j$. Then, 
\begin{align*}
o_A&=c(A')\le\Sigma A' \tag{$\ast$}\\
   &=\Sigma_i(\Sigma\lambda pq\langle
	 1,c(A'),A'(i,q)\rangle)\tag{$\ast\ast$} \\
	 &\le\Sigma_i1=r(A')=l_A,
\end{align*}	 
where the inequality in the line labeled by $(\ast)$ can be shown by
induction on the number of columns of a matrix which has the condition
that each column contains at least one~1; and the equality labeled
with $(\ast\ast)$ follows from the fact that we can add all the
entries in a matrix by rows (and $A'$ is such that each row contains
at most one~1).
\end{proof}

We briefly discuss the implications of Claim~\ref{clm:3} for the
provability of variants of the pigeonhole principle (PHP) in $\LA$ in
Section~\ref{sec:future}.

As Claim~\ref{clm:3} shows, $\LA$ is sufficient to prove $o_A\le l_A$;
on the other hand, we seem to require the stronger theory $\ELA$
(which is $\LA$ with induction over $\Sigma_1^B$ formulas) in order to
prove the other direction of the inequality.

\begin{clm}\label{clm:4}
$\ELA\vdash o_A\ge l_A$.
\end{clm}

\begin{proof}
By induction on $n=|A|$. We assume throughout the proof that the
matrix has the diagonal property (see Definition~\ref{dfn:dc}). Let
\begin{equation}\label{eq:submatrix}
A=\left[\begin{array}{c|ccc}
a & & R & \\\hline
&&& \\
S & & M & \\
&&&
\end{array}\right],
\end{equation}
where $a$ is the top-left entry, and $M$ the principal sub-matrix of
$A$, and $R$ (resp.\ $S$) is $1\times(n-1)$ (resp.\ $(n-1)\times 1$). 
From the diagonal property we know that one of
the following two cases is true:

{\bf Case 1.} $a=1$

{\bf Case 2.} $a,R,S$ consist entirely of zeros

In the second case, $o_A\ge l_A$ follows directly from the induction
hypothesis, $o_M\ge l_M$, as $o_A=o_M\ge l_M=l_A$. Thus, it is the
first case, $a=1$, that is interesting. The first case, in turn, can
be broken up into two subcases: $l_M=n-1$ and $l_M<n-1$.

{\bf Subcase (1-a)} $l_M=n-1$

By induction hypothesis, $o_M\ge l_M=n-1$. We also have that $a=1$,
and $a$ is in position $(1,1)$, and hence no matter what subset of 1s
is selected from $M$, none of them lie on the same line as $a$.
Therefore, $o_A\ge o_M+1$. Since $o_M\ge n-1$, $o_A\ge n$, and since
we can {\em always} cover $A$ with $n$ lines, we have that 
$n\ge l_A$, and so $o_A\ge l_A$.

{\bf Subcase (1-b)} $l_M<n-1$

Consider a covering of $M$ of size $l_M<n-1$. We break this case down
into two further sub-subcases, depending on whether this particular
covering has, or has not, the following property: when all lines of
the covering (of $M$) are extended to the entire matrix $A$, they
cover all the 1s in $S$, or they cover all the 1s in $R$.

For the sake of formalizing the proof in $\LA$, we define the
notion of ``extension'' more precisely.

\begin{dfn}\label{dfn:extension}
Let $A$ and $M$ be as in~(\ref{eq:submatrix}), and let $C_M$ be a set
of lines of $M$, i.e., $C_M$ consists of rows $i_1,i_2,\ldots,i_k$,
and columns $j_1,j_2,\ldots,j_\ell$.  The {\em extension} of $C_M$ to
$C_A$ is simply the set of rows $i_1+1,i_2+1,\ldots,i_k+1$, and the
set of columns $j_1+1,j_2+1,\ldots,j_\ell+1$.
\end{dfn}

{\bf Sub-subcase (1-b-i)}
The given covering of $M$ of size $l_M$, when extended to the full
matrix $A$, covers all the 1s in $S$, or covers all the 1s in $R$ (or
possibly both). 

By induction hypothesis, $o_M=l_M$, and so we can pick $l_M$ 1s in
$M$, no two of them on the same line, plus $a$, to have a selection of
1s, no two on the same line, of size $l_M+1$. Thus, $o_A\ge l_M+1$. On
the other hand, there is a covering of $A$ consisting of the lines
covering $M$ (now extended to all of $A$), plus the first column of
$A$ (if it was $R$ that was fully covered by the extension), or the
first row of $A$ (if it was $S$ that was fully covered by the
extension). Note that if both $R,S$ were fully covered by the
extension, just pick arbitrarily the first row or the first column of
$A$, as all that matters in this case is to cover $a=1$. Therefore,
$l_M+1\ge l_A$, and so $o_A\ge l_A$.

{\bf Sub-subcase (1-b-ii)}
The given covering of $M$, of size $l_M$, 
when extended to $A$, it leaves some 1 in $R$ uncovered,
and some $1$ in $S$ uncovered. 

In that case, $o_A\ge o_M+2$, where we
picked a covering of $M$, extended it to $A$, and picked a selection
of 1s of size $o_M$, no two on the same line, plus a 1 in $R$
uncovered by the extension, and a 1 in $S$ uncovered by the extension,
to create a selection of 1s in $A$ of size $o_M+2$, no two on the same
line.  On the other hand, $l_M+2$ lines cover all of $A$: the
extension of the cover of $M$ of size $l_M$ plus the first row and
first column of $A$. Thus $o_M+2=l_M+2\ge l_A$. Altogether, $o_A\ge l_A$.

This ends the proof of Claim~\ref{clm:4}.
\end{proof}

\section{Induced Algorithm}

The standard KMM Theorem is stated as an implication
(see~(\ref{eq:KMM})), and hence it makes no assertions about the
actual existence of a minimal covering or maximal selection of 1s, let
alone how to compute them. It only says that if they do exist, they
are equal. However, the proof of Lemma~\ref{lem:4} suggests an
algorithm for computing both.

Note that computing a minimal cover can be
accomplished in polytime with the well-known Karp-Hopcroft (KH) algorithm
(see~\cite{hopcroft73}) as follows: First use the KH
algorithm to compute a ``maximal matching,'' which in this case is
simply a maximal selection of~1s (when we view $A$ --- in the natural way
--- as the adjacency matrix of a bipartite graph).
In~\cite{soltys-fernandez-2012}, the authors show how to convert, in
linear time, a maximal selection into a minimal cover. 

Certainly the correctness of the algorithms mentioned in the above
paragraph can be shown in $\ELA$ (as it captures polytime reasoning
--- see~\cite{solcoo}), and so it follows that we can prove in $\ELA$
the existence of a minimal cover and maximum selection.  Therefore,
$\ELA$ can prove something stronger than~(\ref{eq:KMM}).  Namely, it
can not only show that {\em if} we have a minimal cover and a maximal
selection, then they have the same size, but rather, that there {\em
always exists} a minimal cover and maximal selection, and the two
are of equal size. 

However, instead of doing the heavy lifting necessary to formalize the
correctness of HK and~\cite{soltys-fernandez-2012} in $\ELA$, we present a
new simple polytime algorithm for computing minimal covers based on
the proof of Lemma~\ref{lem:4}. Note that a similar argument would
show the existence of a polytime algorithm for maximal selection ---
we leave that to the reader.

The algorithm works as follows: given a 0-1 matrix $A$, we first put
$A$ in the diagonal form (see Definition~\ref{dfn:dc}). We now work
with $A$ which is assumed to be in diagonal form and proceed by
computing recursively $l_M$, the size of a minimal cover of $M$, where
$M$ is the principal submatrix of $A$.  Keeping in mind the form of
$A$ given by~(\ref{eq:submatrix}), we have the following cases:

{\bf Case 1.} If $a=0$ (in which case $R,S$ are also zero, by the fact
that $A$ has been put in diagonal form), then $l_A=l_M$, and proceed
to compute the minimal cover $C$ of $M$; output $C'$, the extension of
$C$ (see Definition~\ref{dfn:extension}).

{\bf Case 2.} If $a\neq 0$, we first examine $R$ to see if the matrix
$M'$, consisting of the columns of $M$ minus those columns of $M$ which
correspond to 1s in $R$, has a cover of size $l_M-\Sigma R$ (of
course, if $l_M<\Sigma R$, then the answer is ``no''). 

If the answer is ``yes'', compute the minimal cover of $M'$, $C_{M'}$.
Then let $C_M$ be the cover of $M$ consisting of the lines in $C_{M'}$
properly renamed to account for the deletion of columns that
transformed $M$ into $M'$, plus the columns of $M$ corresponding the
the 1s in $R$. Then, $C_A$ is the result of extending $C_M$ and adding
the first column of $A$.

If the answer is ``no'', repeat the same with $S$: let $M'$ be the
result of subtracting from $M$ the rows corresponding to the rows with
1s in $S$. Check whether $M'$ has a cover of size $l_M-\Sigma S$. If
the answer is ``yes'' then build a cover for $A$ as in the $R$-case. 

If the answer is ``no'', then compute any minimal cover for $M$,
extend it to $A$, and add the first row and column of $A$; this
results in a cover for $A$.

At the end, we apply the permutations $P,Q$ that converted $A$ to the
diagonal form, to the final $C$, and output that as the minimal cover
of the original $A$. As was mentioned above, a similar polytime
recursive algorithm can compute a maximal selection of 1s; we leave
that to the reader.

\section{Related theorems}

In this section we are going to prove that the various reformulations
of KMM, arising in graph theory and partial orders, can be proven
equivalent to KMM in low complexity ($\LA$), and therefore they also
have feasible proofs. We state this as the following theorem:

\begin{thm}\label{thm:equivalences}
The theory $\LA$ proves the equivalence of KMM, Menger's (restricted),
Hall's and Dilworth's Theorems.
\end{thm}

The proof of consists of
Lemmas~\ref{lem:Menger-to-KMM} and~\ref{lem:KMM-to-Menger}, showing
the equivalence of KMM and a restricted version of Menger's
Theorem in Subsection~A;
Lemmas~\ref{lem:KMM-to-Hall} and
\ref{lem:Hall-to-KMM}, showing the equivalence of KMM and Hall's
Theorem in Subsection~B;
Lemmas~\ref{lem:KMM-to-Dilworth} and~\ref{lem:Dilworth-to-KMM}, 
showing the equivalence of KMM and
Dilworth's Theorem in Subsection~C. Each subsection starts with a
description of how to formalize the given Theorem, followed by the two
Lemmas giving the two directions of the equivalence.

\subsection{Menger's Theorem}\label{sec:menger}

Given a graph $G=(V,E)$, an $x,y$-{\em path} in $G$ is a sequence of
distinct vertices $v_1,v_2, \dots, v_n$ such that $x=v_1$ and $y=v_n$
and for all $1\leq i <n, (v_i,v_{i+1})\in E$.  The vertices
$\{v_2,\dots,v_{n-1}\}$ are called {\em internal vertices}; we say
that two $x,y$-paths are {\em internally disjoint} if they do not have
internal vertices in common.

Given two distinct vertices $x,y\in V$, we say that $S\subseteq E$ is
an {\em $x,y$-cut} if there is no path from $x$ to $y$ in the graph
$G'=(V,E-S)$.  Let $\kappa(x,y)$ represent the size of the smallest
$x,y$-cut, and let $\lambda(x,y)$ represent the size of the largest
set of pairwise internally disjoint $x,y$-paths. 

Menger's theorem states that for any graph $G=(V,E)$, if $x,y \in V$
and $(x,y) \notin E$, then the minimum size of an $x,y$-cut equals the
maximum number of pairwise internally disjoint $x,y$-paths. That is,
$\kappa(x,y)=\lambda(x,y)$. For more details on Menger's Theorem turn
to~\cite{Menger-1927,Goring-2000,pym-1969}. 
Menger's Theorem is of course the familiar Min-Cut Max-Flow Theorem
where all edges have capacity~1.

As was noted earlier, we do not show the equivalence of KMM with the
standard version of Menger's Theorem, but rather with a {\em
restricted} version. Since this restriction is important, we give it
in a definition.

\begin{dfn}\label{dfn:restricted}
Given a graph $G=(V,E)$, we say that a pair of vertices $x,y\in V$ is
{\em restricted} if any $x,y$-path shares edges with at most one other
$x,y$-path.
\end{dfn}

The intuition behind this definition is that given an $x,y$-restricted
pair, there is little ``redundancy'' in the paths between $x$ and $y$. 

We now show how to state Menger's theorem in $\Lang$. We start by
defining the $\Sigma_0^B$ predicate $\Path(A,x,y,\alpha)$, which
states that $\alpha$ encodes the internal vertices of a path from $x$
to $y$ in $A$. We define Path by parts; first we state that $\alpha$
has at most one~1 in each row and column:
\begin{equation}\label{eq:path_part1}
\begin{split}
(\forall l &\leq n-2)
[\Sigma \lambda_{ij}\langle 1,n-2,\alpha(l,j)\rangle =1 \\
&\wedge \Sigma \lambda_{ij}\langle n-2,1,\alpha(i,l)\rangle =1]
\end{split}
\end{equation}
Then we say that if the $l$-th node is $p$ and $l+1$-th node is $q$,
then there is an edge between $p$ and $q$:
\begin{equation}\label{eq:path_part2}
\begin{split}
(\forall &l,p,q \leq n-3) \\ 
&(\alpha(l,p)=1 \wedge \alpha(l+1,q)=1) \ra A(p,q)=1
\end{split}
\end{equation}
Note that in general different paths are of different lengths; this
can be dealt with in a number of ways: for example, by padding
$\alpha$ with repetitions of the last row (so that each $\alpha$ has
exactly $n-2$ rows).  We assume that this is what we do, and the
reader can check that $\mathcal{L}_{\LA}$ can express this easily.

If $i$ is the first intermediate node then $(x,i) \in E$, and if $i$
is the last intermediate node then $(i,y) \in E$:
\begin{equation}\label{eq:path_part3}
\begin{split}
& \alpha(1,i)=1 \ra A(x,i)=1 \\
& \wedge\alpha(n-2,i)=1 \ra A(i,y)=1
\end{split}
\end{equation}

Putting it all together, the $\Sigma_0^B$ formula expressing Path is
given by the conjunction of $A(x,y)=0$ together with the above
properties, i.e.,
\begin{equation}\label{eq:path}
\Path(A,x,y,\alpha):= (\ref{eq:path_part1}) \wedge 
(\ref{eq:path_part2}) \wedge
(\ref{eq:path_part3}) \wedge
A(x,y)=0. 
\end{equation}
Finally, we state that two paths $\alpha,\alpha'$ are internally
disjoint:
\begin{equation}\label{eq:disj_paths}
\begin{split}
&\Disjoint(A,x,y,\alpha,\alpha'):= \\
& \quad\Path(A,x,y,\alpha) \wedge \Path(A,x,y,\alpha')\\
& \quad\wedge (\forall i\leq n-2 \forall j\leq n-2) 
(\alpha(i,j) \cdot \alpha'(i,j)=0)
\end{split}
\end{equation}
We leave stating that $x,y$ is a restricted pair to the reader.

We must be able to talk about a collection of paths; the 0-1 matrix
$\beta$ will encode a collection of paths
$\alpha_1,\alpha_2,\ldots,\alpha_\lambda$:
\begin{equation}\label{eq:collection}
\beta=\begin{array}{|c|c|c|c|c|}
\hline
\beta[1]=\alpha_1 & 
\beta[2]=\alpha_2 & 
\dots & 
\beta[\lambda]=\alpha_\lambda 
\\\hline
\end{array}
\end{equation}
so that $\beta$ is a matrix of size $(n-2)\times \lambda(n-2)$.
Each $\beta[i]$ can be defined thus:
$$
\beta[i]:= \lambda_{pq}\langle n-2,n-2,\beta(p, (i-1)(n-2) + q) \rangle.
$$
We are interested in pairwise disjoint collections of paths:
\begin{equation}\label{eq:coll_disj_paths}
\begin{split}
& \Collectdisj(A,x,y,\beta,\lambda):=\\
& \forall i\leq \lambda \ \Path(A,x,y,\beta[i]) \ \wedge \\
& (\forall i\neq j \leq \lambda ) \ \Disjoint(A,x,y,\beta[i],\beta[j])
\end{split}
\end{equation}
The following formula expresses $\lambda(x,y)$ for a given $A$; note
that it is a $\Pi_2^B$ formula:
\begin{equation}\label{eq:MaxDisj}
\begin{split}
&\Maxdisj(A,x,y,\lambda):= \\ 
&(\exists \beta \leq (n-2)\lambda) \ 
\Collectdisj(A,x,y,\beta,\lambda) \ \wedge \\
& (\forall \alpha \leq n-2) (\Path(A,x,y,\alpha) \ra 
\exists i\leq \lambda \ \alpha=\beta[i]) 
\end{split}
\end{equation}

Likewise, we need to formalize $\kappa(x,y)$; we start with a 0-1
matrix $\gamma$ expressing a cut in $A$:
\begin{equation}\label{eq:cut}
\begin{split}
&\Cut(A,\gamma):=\\
&(\forall i \leq n-2)(\forall j\leq n-2) (\gamma(i,j)=1 \ra A(i,j)=1) 
\end{split}
\end{equation}
which says that every edge of $\gamma$ is an edge of $A$, and it
defines the cut implicitly as the set of edges in $A$ but no in
$\gamma$.  Now, the following $\Sigma_2^B$ formula expresses that
there is an $x,y$-cut of size $\kappa$ in $A$:
\begin{equation}\label{eq:cut_size}
\begin{split}
& \Cutsize(A,x,y,\kappa):= \\
&  \exists \gamma \leq (n-2) \Cut(A,\gamma) \wedge \Sigma\gamma=\kappa
\wedge (\forall \alpha \leq n-2)\\
& \neg \Path(\lambda_{pq}\langle n-2,n-2, A(p,q) - 
\gamma(p,q)\rangle,x,y,\alpha),
\end{split}
\end{equation}
and the minimum number of edges in an $x,y$-cut can be expressed
with a formula that is a conjunction of a $\Sigma_2^B$ formula with a
$\Pi_2^B$ formula, yielding therefore a formula
in $\Sigma_3^B\cap\Pi_3^B$:
\begin{equation}\label{eq:mincut}
\begin{split}
& \Mincut(A,x,y,\kappa):= \\
& \Cutsize(A,x,y,\kappa) \wedge \neg \Cutsize(A,x,y,\kappa-1) 
\end{split}
\end{equation}
Putting it all together, we can state Menger's theorem in $\Lang$
with a $\Sigma_3^B\cap\Pi_3^B$ formula as follows:
\begin{equation}\label{eq:mengers}
\begin{split}
& \Menger(A):=\\
& \Maxdisj(A,x,y,\lambda) \wedge \Mincut(A,x,y,\kappa)
 \rightarrow \lambda = \kappa
\end{split}
\end{equation}
(Note that if a formula is in $\Sigma_3^B\cap\Pi_3^B$, then its
negation is still in $\Sigma_3^B\cap\Pi_3^B$.)

Let $\Menger'$ be the restricted version of Menger's Theorem, i.e.,
one where $x,y$ is a restricted pair, as in
Definition~\ref{dfn:restricted}.  We can now state the main result of
this section.

\begin{lem}\label{lem:Menger-to-KMM}
$\LA\cup\Menger'\vdash\KMM$.
\end{lem}

\begin{proof}
Note that the implication resembles the statement of KMM, but the
difference is that in KMM the antecedent is a conjunction of two
$\Pi_1^B$ formulas (and hence it is a $\Pi_1^B$ formula), whereas in
Menger's theorem, the antecedent is a $\Sigma_3^B\cap\Pi_3^B$ formula.

Suppose
that we have $\Mincover(A,\alpha)\wedge\Maxselect(A,\beta)$, the
antecedent of KMM (see~\ref{eq:KMM}). Using Menger's theorem
(see~\ref{eq:mengers}) we want to conclude that
$\Sigma\alpha=\Sigma\beta$. We do so by restating ``covers and
selections'' of $A$ as ``cuts and paths'' of a related matrix $A'$
defined as: $A'$ is a 0-1 matrix of size
$|A|+1$, with entries:
$$
A'(i,j)=\begin{cases}
A(i,j) & \text{for $1\leq i,j\leq |A|$}\\
1 & \text{one of $\{i,j\}$ equals $|A|+1$}\\
0 & \text{both of $\{i,j\}$ equal $|A|+1$}\\
\end{cases}
$$
Note that $A'$ can be stated succinctly as a term of $\Lang$:
\begin{equation*}\label{eq:df_Aplusxy}
\begin{split}
& A':= \lambda{ij}\langle r(A)+1, c(A)+1,\\
& \cond(1\leq i,j\leq |A|, A(i,j), \cond(i=j=|A|+1,0,1))\rangle
\end{split}
\end{equation*}
The point is that when we view $A$ as representing a bipartite graph
(with rows representing $V_1$ and columns representing $V_2$ and
$A(i,j)=1$ iff there is an edge $(i,j)\in V_1\times V_2$), then $A'$
represents a graph where two more vertices are added ($x=|A|+1$ and
$y=|A|+1$, the first row and column of $A$, resp.) and $x$ is
connected to every vertex in $V_2$ and $y$ is connected to every
vertex in $V_1$, and $x,y$ are not connected to any other vertices.

Also, a maximal selection in $A$ corresponds to a maximal matching in
the related graph, and a minimal cover in $A$ corresponds to a minimal
cover in the related graph (recall that a cover in a graph is a subset
of vertices so that every edge has at least one end-point in this
subset). Furthermore, a maximal matching in the graph related to $A$
corresponds to a maximal subset of internally disjoint paths in the
graph related to $A'$; similarly, a minimal cover in the graph
related to $A$ corresponds to a minimal cut in the graph related to
$A'$.

Finally, let:
$$
A''=\left[\begin{array}{cc}0 & A' \\ (A')^T & 0 \end{array}\right],
$$
that is, $A''$ is the adjacency matrix of the graph related to $A'$
viewed as a normal graph (i.e., not bipartite).

\begin{clm}\label{clm:equivs}
$\LA$ proves the following:
\begin{itemize}
\item $\Mincover(A,\alpha)\lra\Mincut(A'',x,y,\Sigma\alpha)$
\item $\Maxselect(A,\beta)\lra\Maxdisj(A'',x,y,\Sigma\beta)$
\end{itemize}
\end{clm}

The proof of Claim~\ref{clm:equivs} does not require
induction and we leave it to the reader. By Menger's Theorem it
follows directly that $\Sigma\alpha=\Sigma\beta$ which also finishes
the proof of KMM.

It is easy to check that $x,y$ is a restricted pair
(Definition~\ref{dfn:restricted}) in the graph related
to $A''$.
\end{proof}

\begin{lem}\label{lem:KMM-to-Menger}
$\LA\cup\KMM\vdash\Menger'$.
\end{lem}

\begin{proof}
Suppose that we have $\Maxdisj(A,x,y,\lambda)$ and
$\Mincut(A,x,y,\kappa)$; these two formulas assert the existence of
$\beta$, a collection of $\lambda$ many pairwise disjoint $x,y$-paths,
and $\gamma$, an $x,y$-cut of size $\kappa$. (The constructions of
$\beta$ and $\gamma$ have been shown earlier in this section.) We
assume that $x,y$ is a restricted pair of vertices, as in
Definition~\ref{dfn:restricted}.

Each path in $\beta$ must have at least one edge in the cut $\gamma$
and no edge of $\gamma$ can be in more than one path in $\beta$, hence
$\lambda\le\kappa$. The proof of this is identical to the proof of
Claim~\ref{clm:3}. 

Thus, it remains to show, using KMM, that $\lambda\ge\kappa$.  To this
end we proceed as follows: we construct a new matrix $A'$, such that
each row of $A'$ corresponds to one of the paths in~$\beta$. Note that
since the paths in $\beta$ are disjoint, the number of rows of $A'$ is
polynomial in the size of $A$; this is a key observation --- there can
be at most linearly many (in the number of vertices) disjoint paths in
a given graph. On the other hand, the columns of $A'$ correspond to
the edges in $\gamma$, again bounded by a polynomial in the size of
$A$, as there are at most $|A|^2$ edges in the graph.  We have
$A'(i,j)=1\iff$ edge $j$ is in the path $i$.

Note that $\beta$ and $\gamma$ are built independently; the only
assertion we make about their relationship is that they are of the
same size, i.e., $\lambda=\kappa$. In the next Claim we show how we
can modify $\beta$ and $\gamma$ (using an algorithm provably correct
in $\LA$) in order to obtain an $A'$ to which we can apply KMM.

\begin{clm}\label{clm:7}
We can modify $\beta$ and $\gamma$, with a procedure provably correct
in $\LA$, so that each row and column of $A'$ contains exactly one~1.
\end{clm}

Before we prove Claim~\ref{clm:7} we show how we use it to show that
$\lambda\ge\kappa$: if each row and column of $A'$ has exactly one~1,
then $\kappa=o_{A'}$, and by KMM, $o_{A'}=l_{A'}\le
r(A')=\lambda$.

We now prove Claim~\ref{clm:7}.  First observe that no matter what
$\beta$ and $\gamma$ we pick, each column of $A'$ has at most one~1,
and each row of $A'$ has at least one~1.  The reason is that the paths
in $\beta$ are pairwise disjoint --- hence they never share an edge.
If a row of $A'$ contains no~1s, then we have an $x,y$-path, and
$\gamma$ is not an $x,y$-cut.

Suppose there is a column without a~1. Then there is an edge
$e\in\gamma$ that
does not belong to any path in $\beta$; if $e$ were unnecessary, we
could lower $\kappa$, and get a contradiction (with the minimality of
$\kappa$). Thus, $e$ cuts some $x,y$-path $\rho$ not in $\beta$; by
the maximality of $\beta$, $\rho$ shares an edge $e'$ with some path
$\rho'$ which is in $\beta$. If $e'$ is in the cut $\gamma$, then $e$
is not needed --- contradiction again. So $e'$ is not in $\gamma$;
exchange $e$ and $e'$ in order to obtain a new $\gamma$. Here is were
we use the restriction on the graphs, i.e., we know that $e$ is not
shared by any other paths, as $\rho$ shares its edge with $\rho'$. We
leave showing that each row has at most one~1 to the reader.
\end{proof}

\subsection{Hall's Theorem}

Let $S_1,S_2,\dots,S_n$ be $n$ subsets of a given set $M$. Let $D$ be
a set of $n$ elements of $M$, $D=\{a_1,a_2,\dots,a_n\}$, such that
$a_i \in S_i$ for each $i=1,2,\dots,n$. Then $D$ is said to be a {\em
system of distinct representative} (SDR) for the subsets
$S_1,S_2,\dots,S_n$.

If the subsets $S_1,S_2,\dots,S_n$ have an SDR, then any $k$ of the
sets must contain between them at least $k$ elements. The converse
proposition is the combinatorial theorem of P.\ Hall: suppose that for
any $k=1,2,\ldots,n$, any $S_{i_1}\cup S_{i_2}\cup\cdots\cup S_{i_k}$
contains at least $k$ elements of $M$; we call this the {\em union
property}.  Then there exists an SDR for these subsets.
(See~\cite{hall-1987,everett-1949,halmos-1950} for more on Hall's
theorem.)

We formalize Hall's theorem in $\Lang$ with an adjacency matrix $A$
such that the rows of $A$ represent the sets $S_i$, and the columns of
$A$ represent the indices of the elements in $M$, i.e., the columns
are labeled with $[n]=\{1,2,\ldots,n\}$, and $A(i,j)=1\iff j\in S_i$.
Let $\SDR(A)$ be the following $\Sigma_1^B$ formula which states that
$A$ has a system of distinct representatives:
\begin{equation}\label{eq:existsdr}
\SDR(A):=(\exists P\leq n)(\forall i\leq n)(AP)_{ii}=1
\end{equation}
We reserve the
letters $P,Q$ for permutation matrices, and $(\exists P\le
n)\phi$ abbreviates $(\exists P)[\Perm(P)\wedge |P|\le n\wedge\phi]$
(similarly for $(\forall P\le n)\phi$, but with an implication
instead of a conjunction),
where $\Perm$ is a $\Sigma_0^B$ predicate stating that $P$ is a
permutation matrix (a unique~1 in each row and column).
See~\cite{soltys-tcs04} for more details about handling permutation
matrices. 

The next
predicate is a $\Pi_2^B$ formula stating the union property:
\begin{equation}\label{eq:unionprop-sdr}
\begin{split}
&\Unionprop(A):=  \\
& \forall P\le n \forall k\leq n 
\exists Q\le n \\
& [\forall i\leq k (\lambda_{pq}\langle k,1,(PAQ)_{pi}\rangle 
\neq \lambda_{pq}\langle k,1,0 \rangle)]
\end{split}
\end{equation}

Therefore, we can state Hall's theorem as a $\Sigma_2^B$ formula:
\begin{equation}\label{eq:HallsTh}
\Hall(A):= \Unionprop(A) \rightarrow \SDR(A)
\end{equation}

\begin{lem}\label{lem:KMM-to-Hall}
$\LA \cup \KMM \vdash \Hall$.
\end{lem}

\begin{proof}
Let $A$ be a 0-1 sets/elements incidence matrix of size $n \times n$.
Assume that we have $\Unionprop(A)$; our goal is to show in $\LA$,
using~KMM, that $\SDR(A)$ holds. 

Since by Claim~\ref{clm:2}, every matrix can be put in a diagonal
form, using the fact that we have $\Unionprop(A)$, it follows that we
can find $P,Q\le n$ such that $\forall k\leq n (PAQ)_{kk}=1$.  Thus we
need $n$ lines to cover all the~1s, but by $\KMM$ there exists a
selection of $n$ 1s no two on the same line, hence, $A$ is of term
rank $n$.  

But this means that the maximal selection of 1s, no two on the same
line, constitutes a permutation matrix $P$ (since $A$ is $n\times n$,
and we have $n$ 1s, no two on the same line). Note that $AP^T$ has all
ones on the diagonal, and this in turn implies $\SDR(A)$.
\end{proof}

\begin{lem}\label{lem:Hall-to-KMM}
$\LA \cup \Hall \vdash \KMM$.
\end{lem}

\begin{proof}
Suppose that we have $\Mincover(A,\alpha)$ and $\Maxselect(A,\beta)$;
we want to conclude that $\Sigma\alpha=\Sigma\beta$ using Hall's
Theorem.

As usual, let $l_A=\Sigma\alpha$ and $o_A=\Sigma\beta$, and
by Claim~\ref{clm:3} we already have that
$\LA\vdash o_A\le l_A$ (see Section~\ref{sec:II}).  
We now show in $\LA$ that $o_A\ge l_A$ using Hall's Theorem.  

Suppose
that the minimum number of lines that cover all the 1s of $A$ consists
of~$e$ rows and~$f$ columns, so that $l_A=e+f$. Both $l_A$ and $o_A$
are invariant under permutations of the rows and the columns of $A$
(Lemma~\ref{lem:1}), and so we reorder the rows and columns of $A$ so
that these~$e$ rows and~$f$ columns are the initial rows and columns
of $A'$,
$$
A'=\left[\begin{array}{cc}
A_1 & A_2 \\
A_3 & A_4 
\end{array}\right],
$$
where $A_1$ is of size $e \times f$.  Now, we shall work with the term
rank of $A_2$ and $A_3$ in order to show that $o_A\ge l_A$.  More
precisely, we will show that the maximum number of 1s, no two on the
same line, in $A_2$ is $e$, while in $A_3$ it is $f$.

Let us consider $A_2$ as an incidence matrix for subsets
$S_1,S_2,\ldots, S_e$ of a universe of size $|A|-f$, and $A_3^t$
(which is the transpose of $A_3$) as an incidence matrix for subsets
$S'_1,S'_2,\ldots,S'_f$ of a universe of size $|A|-e$.
It is not difficult to prove that $\Unionprop(A_2)$ and
$\Unionprop(A^t_3)$ holds (and can be proven in $\LA$; this is left to
the reader), which in turn
implies $\SDR(A_2)$ and $\SDR(A^t_3)$, resp., by Hall's Theorem.  But
the system of distinct representative of $A_2$ (resp.\ $A^t_3$) implies that
$o_{A_2}\ge e$ (resp.\ $o_{A^t_3}=o_{A_3}\ge f$), and since $o_A\ge
o_{A_2}+o_{A_3}$, this yields that
$o_A\ge e+f=l_A$.
\end{proof}

\subsection{Dilworth's Theorem}

Let $\mathscr{P}$ be a {\em finite partially ordered set} or {\em
poset} (we use a ``script $\mathscr{P}$'' in order to distinguish it
from permutation matrices, denoted with $P$).  We say that $a,b \in
\mathscr{P}$ are {\em comparable elements} if either $a<b$ or $b<a$. A
subset $C$ of $\mathscr{P}$ is a {\em chain} if any two distinct
elements of $C$ are comparable. A subset $S$ of $\mathscr{P}$ is an
{\em anti-chain} (also called an {\em independent set}) if no two
elements of $S$ are comparable.

We want to partition a poset into chains; a poset with an anti-chain of
size $k$ cannot be partitioned into fewer than $k$ chains, because any
two elements of the anti-chain must be in a different partition.
Dilworth's Theorem states that the maximum size of an anti-chain equals
the minimum number of chains needed to partition $\mathscr{P}$.  (For
more on Dilworth's Theorem see~\cite{dilworth-1950,perles-1963}). 

In order to formalize Dilworth's theorem in $\Lang$, we represent
finite posets $\mathscr{P}=(X=\{x_1,x_2,\ldots,x_n\},<)$ with an
incidence matrix $A=A_\mathscr{P}$ of size $|X|\times|X|$,
which expresses the relation $<$ as follows: $A(i,j)=1\iff x_i<x_j$.
For more material regarding formalizing posets
see~\cite{soltys-szpilrajn}.

We let a $1\times n$ matrix $\alpha$ encode a chain as follows:
\begin{equation}\label{eq:chain}
\begin{split}
\Chain&(A,\alpha):= (\forall i\neq j\le n)\\
& [\alpha(i)=\alpha(j)=1\ra A(i,j)=1\vee A(j,i)=1].
\end{split}
\end{equation}
In a similar fashion to~(\ref{eq:chain}) 
we define an anti-chain $\gamma$; the only difference is that the
succedent of the implication expresses the opposite:
$A(i,j)=0\wedge A(j,i)=0$.

Recall that using~(\ref{eq:collection}) we were able to talk about a
collection of paths; in a similar vain, we can use $\Lang$ to talk
about a collection of chains of $\mathscr{P}$: $\beta$ is an $1\times
\kappa\cdot n$ matrix which encodes the contents of $\kappa$ many
chains.  We can then talk about a minimal collection of chains, or a
maximal size of an anti-chain in the usual fashion.  Since we have
done this already for collections of paths, we omit the details in the
interest of space. The reader is encouraged to fill in the details.

We can state Dilworth's Theorem as follows:
\begin{equation}\label{eq:DilTh}
\begin{split}
& \Dilworth(A):= (\exists\beta\le|A|^2)(\exists\gamma\le|A|)\\
& \qquad\Minchain(A,\beta,\kappa)\wedge\Maxantichain(A,\gamma,\lambda)
\rightarrow \lambda = \kappa
\end{split}
\end{equation}
where the predicate 
$\Minchain(A,\beta,\kappa)$ asserts that $\beta$ is a collection
of $\kappa$ many chains that partition the poset, and
the predicate
$\Maxantichain(A,\gamma,\lambda)$ asserts that $\gamma$ is an
anti-chain consisting of $\lambda$ elements. Again, the details of the
$\Lang$ definitions can be provided by the reader, in light of the
definitions given in Section~\ref{sec:menger}.

\begin{lem}\label{lem:KMM-to-Dilworth}
$\LA \cup \KMM \vdash \Dilworth$
\end{lem}

\begin{proof}
Suppose that we have $\Minchain(A,\beta,\kappa)$ and
$\Maxantichain(A,\gamma,\lambda)$;
we want to use $\LA$ reasoning and KMM in order
to show that $\lambda=\kappa$.

As usual we define a matrix $A'$ whose rows are labeled by the chains
in $\beta$, and whose columns are labeled by the elements of the
poset. As there cannot be more chains than elements in the poset, it
follows that the number of rows of $A'$ is bounded by $|A|$ (while the
number of columns is exactly~$|A|$). The proof of this is similar to
the proof of Claim~\ref{clm:3}.

We have that $A'(i,j)=1\iff$ chain $i$ contains element~$j$.
Clearly each column contains at least one~1, as $\beta$ is a partition
of the poset. On the other hand, rows may contain more than one~1, as
in general chains may have more than one element.

Note that a maximal selection of 1s, no two of them on the same line,
corresponds naturally to a maximal anti-chain; such a selection picks
one~1 from each line, and so its size is the number of rows of $A'$.
By KMM, it follows that 
$$
\lambda=o_{A'}=l_{A'}=r(A')=\kappa,
$$
where $r(A')$ is the number of rows of $A'$.
\end{proof}

\begin{lem}\label{lem:Dilworth-to-KMM}
$\LA \cup \Dilworth \vdash \KMM$
\end{lem}

\begin{proof}
It is in fact easier to show that that $\LA\cup\Dilworth\vdash\Hall$,
and since by Lemma~\ref{lem:Hall-to-KMM} we have that
$\LA\cup\Hall\vdash\KMM$, we will be done.

In order to prove Hall using Dilworth and $\LA$ reasoning, we assume
that we have $A$, a 0-1 sets/elements incidence matrix of size $n
\times n$.  Assume that we have $\Unionprop(A)$; our goal is to show
in $\LA$, using Dilworth, that $\SDR(A)$ holds. 

Let $S_1,S_2,\ldots,S_n$ be subsets of $\{x_1,x_2,\ldots,x_n\}$ where
$n=|A|$. 
We define a partial order $\mathscr{P}$ based on $A$; the universe of
$\mathscr{P}$ is $X=\{S_1,S_2,\ldots,S_n\}\cup\{x_1,\ldots,x_n\}$.
The relation $<_\mathscr{P}$ is defined as follows:
$x_i<_\mathscr{P}S_j\iff A(i,j)=1$.

\begin{clm}
The maximum size of an anti-chain in $\mathscr{P}$ is $n$.
\end{clm}

\begin{proof}
The $\{x_1,\ldots,x_n\}$ form an anti-chain
of length $n$, and we cannot add any of the $S_j$, as some $x_i\in
S_j$, and hence $x_i<_\mathscr{P}S_j$. 
\end{proof}

By Dilworth we can partition $\mathscr{P}$ into $n$ chains, where each
of the chains has two elements $\{x_i,S_j\}$, giving us the set of
distinct representatives, and hence $\SDR(A)$.
\end{proof}

\section{Future work}\label{sec:future}

The main open question is the following: is KMM equivalent (in $\LA$)
to the {\em general} version of Menger's Theorem? That is, can we lift the
restriction that $x,y$ is a restricted pair of vertices
(Definition~\ref{dfn:restricted}).
There are many proofs of the general version of Menger's theorem ---
for example~\cite{Goring-2000} is clearly formalizable in $\ELA$. But
it would be very interesting to know whether the general version of 
Menger's Theorem is
equivalent to KMM in low complexity.

Now that we know that $\ELA\vdash\text{KMM}$ (Theorem~\ref{thm:main})
and that KMM is equivalent to a host of other combinatorial theorems
--- and this equivalence can be shown in the weak theory $\LA$
(Theorem~\ref{thm:equivalences}) --- it would be interesting to know
whether it is also the case that: 
\begin{itemize}
\item $\LA\cup\KMM\vdash\text{PHP}$
\item $\LA\cup\text{PHP}\vdash\KMM$
\end{itemize}
that is, whether $\LA$ can prove the equivalence of KMM and the
pigeonhole principle. We conjecture that the first assertion is true,
and that it should not be too difficult to show it. The second
assertion is probably not true. Note that in the proof of
Claim~\ref{clm:3} we implicitly show a certain weaker kind of the PHP
in $\LA$:
we showed that if we have a set of $n$ items $\{i_1,i_2,\ldots,i_n\}$
and a second set of $m$ items $\{j_1,j_2,\ldots,j_m\}$, and we can
match each $i_p$ with some $j_q$, and this matching is both
definable in $\LA$ and its injectivity is provable in $\LA$,
then $n\le m$. We did this by defining an incidence matrix
$A$ such that $A(p,q)=1\iff i_p\mapsto j_q$. If this mapping is
injective, then each column of $A$ has at most one~1; thus:
$$
n\le\Sigma A=\Sigma_i(\text{col $i$ of $A$})\le\Sigma_i1\le m.
$$

Also, we would like to know whether $\LA\cup\KMM$ can prove hard
matrix identities, such as $AB=I\ra BA=I$. Of course, we already know
from~\cite{Tzameret-2011} that (non-uniform) $\textbf{NC}^2$ Frege is
sufficient to prove $AB=I\ra BA=I$. On the other hand, is it possible
that $\LA$ together with $AB=I\ra BA=I$ can prove KMM? This would
imply that $AB=I\ra BA=I$ is ``complete'' for
combinatorial matrix algebra, in the sense that all of combinatorial
matrix algebra follows from this principle with proofs of low
complexity.

Furthermore, given two 0-1 matrices $A,B$, what can we say about
$l_{AB}$ and $o_{AB}$?  From Claim~\ref{lem:1} we know that if $B$ is
a permutation matrix, then $l_{AB}=l_A$ and $o_{AB}=o_A$ (and
similarly, if $A$ is a permutation matrix); but what can be said in
general? Of course, the understanding here is that multiplication is
over the field $\{0,1\}$.

\section{Appendix --- $\LA$}\label{sec:appendix}

The logical theory $\LA$ is strong enough to prove all the ring properties
of matrices such as 
$A(BC)=(AB)C$ and 
$A+B=B+A$,
but weak enough so that the theorems of $\LA$ translate into
propositional tautologies with short Frege proofs.
$\LA$ has three sorts of object: {\em indices} (i.e., natural
numbers), {\em ring elements}, and {\em matrices}, where the
corresponding variables are denoted $i,j,k,\ldots$;  $a,b,c,\ldots$; and
$A,B,C,\ldots$, respectively.   The semantic assumes that objects of
type ring are from a fixed but arbitrary ring (for the purpose of this
paper we are only interested in the ring $\mathbb{Z}$), and
objects of type matrix have entries from that ring.

Terms and formulas are built from the following function and predicate
symbols, which together comprise the language $\mathcal{L}_\LA$:
\begin{equation}\label{symbols}
\begin{split}
&0_{\text{index}},
1_{\text{index}},
+_{\text{index}},
*_{\text{index}},
-_{\text{index}},
\quot,
\rem,  \\
&0_{\text{ring}},
1_{\text{ring}}, 
+_{\text{ring}},
*_{\text{ring}},
-_{\text{ring}}, {}^{-1},
\row, \col, \entry, \Sigma,\\
&\leq_{\text{index}},
=_{\text{index}},
=_{\text{ring}}, 
=_{\text{matrix}},
\cond_{\text{index}}, 
\cond_{\text{ring}}
\end{split}
\end{equation}
The intended meaning should be clear, except in the case of
$-_{\text{index}}$,
cut-off subtraction, defined as $i-j=0$ if $i<j$.
For a matrix $A$: $\row(A), \col(A)$ are the numbers of rows
and columns in $A$, $\entry(A,i,j)$ is the ring element $A_{ij}$
(where $A_{ij}=0$ if $i=0$ or $j=0$ or $i>\row(A)$ or $j>\col(A)$), 
$\Sigma(A)$ is the sum of the elements in $A$.  Also $\cond(\alpha,
t_1,t_2)$ is interpreted {\bf if} $\alpha$ {\bf then} $t_1$ {\bf else}
$t_2$, where $\alpha$ is a formula all of whose atomic sub-formulas
have the form $m\leq n$ or $m=n$, where $m,n$ are terms of type index,
and $t_1,t_2$ are terms either both of type index or both of type
ring.
The subscripts $_{\text{index}}$, $_{\text{ring}}$, and
$_{\text{matrix}}$  are
usually omitted, since they ought to be clear from the context.

We use $n,m$ for terms of type index, $t,u$ for terms of type ring,
and $T,U$ for terms of type matrix.  Terms of all three types are
constructed from variables and the symbols above in the usual way,
except that terms of type matrix are either variables $A,B,C,...$
or $\lambda$-terms $\lambda ij\langle m,n,t\rangle$.
Here $i$ and $j$ are variables of type index bound by the $\lambda$
operator, intended to range over the rows and columns of the matrix.
Also $m,n$
are terms of type index {\em not} containing $i,j$ (representing the
numbers of rows and columns of the matrix) and $t$ is a term of type
ring (representing the matrix element in position $(i,j)$).

Atomic formulas are of the form $m\leq n, m=n, t=u$ and $T=U$, where
the three occurrences of = formally have subscripts $_{\text{index}},
_{\text{ring}}, _{\text{matrix}}$, respectively.  General formulas are
built from atomic formulas using the propositional connectives
$\neg,\vee,\wedge$ and quantifiers $\forall,\exists$.

\subsection{Axioms and rules of $\LA$}

For each axiom listed below, every
legal substitution of terms for free variables is an axiom of $\LA$.
Note that in a $\lambda$ term $\lambda ij\langle m,n,t\rangle$
the variables $i,j$ are bound.  Substitution instances must
respect the usual rules which prevent free variables from
being caught by the binding operator $\lambda ij$.  The bound
variables $i,j$ may be renamed to any new distinct pair of variables.

\subsubsection{Equality Axioms}

These are the usual equality axioms, generalized to apply to the
three-sorted theory $\LA$.
Here = can be any of the three equality symbols, $x,y,z$ are variables
of any of the three sorts (as long as the formulas are syntactically
correct).  In A4, the symbol $f$ can be any of the non-constant function
symbols of $\LA$.  However A5 applies only to $\leq$, since this in the
only predicate symbol of $\LA$ other than =.

\begin{tabbing}
{\bf A1} \hskip 3mm \= $x=x$\\
{\bf A2} \> $x=y\ra y=x$\\
{\bf A3} \> $(x=y\wedge y=z)\ra x=z$\\
{\bf A4} \> $x_1=y_1,...,x_n=y_n\ra fx_1...x_n=fy_1...y_n$\\
{\bf A5} \> $i_1=j_1,i_2=j_2,i_1\leq i_2\ra j_1\leq j_2$
\end{tabbing}

\subsubsection{Axioms for indices}

These are the axioms that govern the behavior of index elements.  The
index elements are used to access the entries of matrices, and so we
need to define some basic number theoretic operations.

\begin{tabbing}
{\bf A6} \hskip 3mm  \= $i+1\not= 0$\\
{\bf A7}   \> $i*(j+1)=(i*j)+i$\\
{\bf A8}   \> $i+1=j+1\ra i=j$\\
{\bf A9}   \> $i\leq i+j$\\
{\bf A10}  \> $i+0=i$\\
{\bf A11}  \> $i\leq j\wedge j\leq i$\\
{\bf A12}  \> $i+(j+1)=(i+j)+1$\\
{\bf A13}  \> $[i\leq j\wedge j\leq i]\ra i=j$\\
{\bf A14}  \> $i*0=0$\\
{\bf A15}  \> $[i\leq j\wedge i+k=j]\ra j-i=k$ \\ 
{\bf A16}  \> $\neg(i\leq j)\ra j-i=0$ \\
{\bf A17}  \> $[\alpha\ra\text{cond}(\alpha,i,j)=i]\wedge
               [\neg\alpha\ra\cond(\alpha,i,j)=j]$
\end{tabbing}

\subsubsection{Axioms for a ring}

These are the axioms that govern the behavior for ring elements;
addition and multiplication, as well as additive inverses.  We do not
need multiplicative inverses.

\begin{tabbing}
{\bf A18} \hskip 3mm \= $0\not= 1 \wedge a+0=a$\\
{\bf A19} \> $a+(-a)=0$\\
{\bf A20} \> $1*a=a$\\
{\bf A21} \> $a+b=b+a$\\
{\bf A22} \> $a*b=b*a$\\
{\bf A23} \> $a+(b+c)=(a+b)+c$\\
{\bf A24} \> $a*(b*c)=(a*b)*c$\\
{\bf A25} \> $a*(b+c)=a*b+a*c$\\
{\bf A26} \> $[\alpha\ra\cond(\alpha,a,b)=a]
              \wedge[\neg\alpha\ra\cond(\alpha,a,b)=b]$
\end{tabbing}

\subsubsection{Axioms for matrices}

Axiom A27 states that
$\entry(A,i,j)$ is zero when $i,j$ are outside the size of $A$.  Axiom
A28 defines the behavior of constructed matrices.
Axioms A29-A32 define the function
$\Sigma$ recursively by first defining it for row vectors,
then column vectors ($A^t:=\lambda
ij\langle\col(A),\row(A),A_{ji}\rangle$),
and then in general using the decomposition (\ref{eq:RSM}).
Finally, axiom A33 takes care of empty matrices.

\begin{tabbing}
{\bf A27} \= $(i=0\vee\row(A)<i\vee j=0\vee\col(A)<j) \ra$ \\
             \> $\ra\entry(A,i,j)=0$\\
{\bf A28} \> $\row(\lambda ij\langle m,n,t\rangle)=m
              \wedge\col(\lambda ij\langle m,n,t\rangle)=n
              \wedge$ \\ \>$[1\leq i\wedge i\leq m\wedge 1\leq j\wedge j\leq n] \ra$ \\ \>$\ra\entry(\lambda ij\langle m,n,t\rangle,i,j)=t$ \\
{\bf A29} \> $\row(A)=1,\col(A)=1\ra\Sigma(A)=\entry(A,1,1)$\\
{\bf A30} \> $\row(A)=1\wedge 1<\col(A)
             \ra\Sigma(A)=$ \\ \>$=
           \Sigma(\lambda ij\langle 1,\col(A)-1,A_{ij}\rangle)+A_{1\col(A)}$\\
{\bf A31} \> $\col(A)=1\ra\Sigma(A)=\Sigma(A^t)$\\
{\bf A32} \> $1<\row(A)\wedge 1<\col(A)
          \ra\Sigma(A)=$\= \\ \>$=\entry(A,1,1)+
          \Sigma(\R(A))+\Sigma(\SA(A))+\Sigma(\M(A))$ \\
{\bf A33} \> $\row(A)=0\vee\col(A)=0\ra\Sigma A=0$
\end{tabbing}

Where
\begin{equation}\label{eq:RSM}
\begin{split}
\R(A)&:=\lambda ij\langle 1,\col(A)-1,\entry(A,1,i+1)\rangle, \\
\SA(A)&:=\lambda ij\langle \row(A)-1,1,\entry(A,i+1,1)\rangle, \\
\M(A)&:=\lambda ij\langle
\row(A)-1,\col(A)-1,\entry(A,i+1,j+1)\rangle.
\end{split}
\end{equation}\label{RSM}

\subsubsection{Rules for $\LA$}

In addition to all the axioms just presented, $\LA$
has two rules:  matrix equality and induction.

\medskip
\noindent
{\bf Matrix equality rule}

From the premises:
$\entry(T,i,j)=\entry(U,i,j)$,
$\row(T)=\row(U)$ and
$\col(T)=\col(U)$,
we conclude $T=U$.

The only restriction is that the variables $i,j$ may not occur free in
$T=U$; other than that, $T$ and $U$ can be arbitrary matrix terms.
Our semantics
implies that $i$ and $j$ are implicitly universally quantified in
the top formula.  The rule allows us to conclude
$T=U$, provided that $T$ and $U$ have the same numbers of rows and
columns, and corresponding entries are equal.

\medskip
\noindent
{\bf Induction rule}
$\alpha(i)\ra\alpha(i+1)$ implies
$\alpha(0)\ra\alpha(n)$.

Here $\alpha(i)$ is any formula, $n$ is any term of type index, and
$\alpha(n)$ indicates $n$ is substituted for free occurrences of $i$
in $\alpha(i)$.  (Similarly for $\alpha(0)$.) 
Note that in $\LA$ we
only allow induction over $\Sigma_0^B$ formulas (no matrix
quantifiers), 
whereas in $\ELA$ we allow induction over $\Sigma_1^B$
formulas 
(a single block \newpage\noindent of bounded existential matrix quantifiers
when $\alpha$ is put in prenex form).
This completes the description of $\LA$.  We finish this section
by observing the substitution property in the lemma below.
We say that a formula $S'$ of $\LA$ is a {\em substitution instance}
of a formula $S$ of $\LA$ provided that $S'$ results by substituting
terms for free variables of $S$.  Of course each term must have
the same sort as the variable it replaces, and bound variables
must be renamed as appropriate.

\begin{lem}
Every substitution instance of a theorem of $\LA$ is a theorem of $\LA$.
\end{lem}

This follows by straightforward
induction on $\LA$ proofs.  The base case follows from
the fact that every substitution instance of an $\LA$ axiom is an
$\LA$ axiom.

\bibliographystyle{IEEEtran}


\end{document}